\documentclass{amsart}

\newtheorem{theorem}{Theorem}[section]
\newtheorem{lemma}[theorem]{Lemma}

\theoremstyle{definition}

\numberwithin{equation}{section}

\usepackage{hyperref}
\begin{document}

\title{Existence of solution to an evolution equation and\\
a justification of the DSM for equations with monotone operators}



\author{N. S. Hoang}
\address{Mathematics Department, Kansas State University,
Manhattan, KS 66506-2602, USA}
\curraddr{}
\email{nguyenhs@math.ksu.edu}
\thanks{}

\author{A. G. Ramm}
\address{Mathematics Department, Kansas State University,
Manhattan, KS 66506-2602, USA}
\curraddr{}
\email{ramm@math.ksu.edu}

\subjclass[2000]{Primary 47J05; 47J06; 47J35}

\date{}

\dedicatory{}

\keywords{Dynamical systems method (DSM),
nonlinear operator equations, monotone operators.}


\begin{abstract} An evolution equation, arising 
in the study of the Dynamical Systems Method (DSM) for solving equations 
with 
monotone operators, is studied in this paper. The evolution equation is 
a continuous analog of the regularized Newton method for solving ill-posed 
problems with monotone nonlinear operators $F$. Local and global 
existence of the 
unique solution to this evolution equation are proved, apparently for the 
firs time, under  
the only assumption that $F'(u)$ exists and is continuous with respect to 
$u$.
The earlier published results required more smoothness of $F$.
The Dynamical Systems method (DSM) for solving equations $F(u)=0$ with
monotone Fr\'echet differentiable operator $F$ 
is justified under the above assumption apparently for the first time. 

\end{abstract}

\maketitle

\section{Introduction}
The Dynamical Systems Method (DSM) for solving an operator equation
$F(u)=f$ in a Hilbert space consists of finding a nonlinear map 
$\Phi(u,t)$ such that the Cauchy problem 
$$\dot{u}=\Phi(t,u),\quad u(0)=u_0;\qquad
\dot{u}:=\frac{du}{dt},$$
has a unique global solution, there exists $\lim_{t\to 
\infty}u(t):=u(\infty), $ and $F(u(\infty))=f$  (see \cite{R499}).
Here $u_0\in H$ is an arbitrary element, possibly belonging to
a bounded subset of $H$.
  
One of the versions of the DSM (\cite{R499})
for solving nonlinear operator equation $F(u)=f$ with monotone operator
$F$ in a Hilbert space is based on a regularized continuous analog of the
Newton method, which consists of solving the 
following Cauchy problem \begin{equation} \label{eq1} 
\dot{u} = -\big{(}F'(u) + a(t)I\big{)}^{-1}\big{(}F(u)+a(t)u - f\big{)}, 
\quad u(0)=u_0. \end{equation} 
Here $F:H\to H$ is a monotone continuously Fr\'echet 
differentiable operator in 
a Hilbert space $H$, $u_0$ and $f$ in $H$ are arbitrary, and $a(t)>0$ is 
a continuously differentiable function, defined for all $t\geq 0$ and 
decaying to zero as $t\to \infty$. This function is a regularizing 
function: if $F'(u)$ is not a boundedly invertible operator, and $f$ is 
monotone, then $F'(u)\geq 0$ and the operator $F'(u)+a(t)I$ is boundedly 
invertible if $a(t)>0$. By $I$ the identity operator is denoted, and 
by $\langle u,v\rangle$ we denote the inner product in $H$. Monotonicity 
of $F$ is understood as follows 
\begin{equation} \label{eqx12} 
\langle F(u) - 
F(w), u - w\rangle \ge 0,\qquad \forall u,w\in H. 
\end{equation} 
The DSM is a basis for developing efficient numerical methods
for solving operator equations, both linear and nonlinear, especially 
when the problems are ill-posed, for example, when $F'(u)$ is not a 
boundedly invertible operator (see \cite{R499}, \cite{R574}).

If one has a general evolution problem with a nonlinear operator
in a Hilbert (or Banach) space
\begin{equation} \label{eqgen}
\dot{u} = B(u), \qquad u(0)=u_0,
 \end{equation}
then the local existence of the solution to this   
problem is usually established by assuming that 
$B(u)$ satisfies a Lipschitz condition, and the global existence is 
usually established by proving a uniform bound on the 
solution:
\begin{equation} \label{eqbound} 
\sup_{t\geq 0}||u(t)||<c,
 \end{equation}
where $c>0$ is a constant.

In \eqref{eq1} the operator 
$$B(u)= -\big{(}F'(u) + 
a(t)I\big{)}^{-1}\big{(}F(u)+a(t)u - f\big{)}$$
is Lipschitz only if one assumes that 
$$\sup_{\{u:  
||u-u_0||\leq R\}}||F^{(j)}(u)||\leq M_j(R), \quad 0\leq j\leq 2.$$ 
This 
assumption was used in many cases in \cite{R499} and 
a bound  \eqref{eqbound} was established under suitable assumptions
in \cite{R499}.

There are many results (see, e.g., \cite{D},\cite{M} and references 
therein) concerning the properties 
and  global
existence of the solution to  \eqref{eqgen} if $-B(u)$ is a maximal 
monotone operator. However, even when $F$ is a monotone
operator, the operator $-B$ in the right-hand side of 
 \eqref{eq1} is not monotone. {\it Therefore these known results
are not applicable. Even the proof of local existence is an open 
problem if one makes only the following assumption}:

{\bf Assumption A): 

 $F$ is monotone and $F'(u)$ is continuous with respect to $u$}. 

The main result of this paper is a proof, apparently published for the 
first time,  that under 
Assumption A)  problem  \eqref{eq1} has a unique local solution
$u(t)$, and that under assumptions \eqref{eq2.2} on $a(t)$
(see below)   this local 
solution exists for all $t\geq 0$, so it is a global solution.
These results are formulated in Theorems 1.1 and 2.2.

Moreover, if the equation $F(y)=f$ has a solution and  $y$ is its (unique) 
minimal-norm solution, and if $\lim_{t\to \infty}a(t)=0$,  then 
there exists $u(\infty)$, and $u(\infty)=y$. This justifies the DSM
for solving the equation $F(u)=0$ with a monotone continuously Fr\'echet 
differentiable operator $F$, for the first time under the weak 
Assumption A). This result is formulated in Theorem 3.1.

Let us prove the existence of the solution to \eqref{eq1}.

Let 
\begin{equation}
\label{eq1.5}
\psi = F(u) +a(t)u - f := G(u,t).
\end{equation}
If $a(t)>0$ and $F$ is monotone and hemicontinuous, then it is known
(see, e.g., \cite{D}, p. 100) that the operator 
$F(u)+a(t)u$ is surjective. If $F'(u)$ is continuous, then, clearly, $F$ 
is hemicontinuous.
If $F$ is monotone and $a(t)>0$ then, clearly, the operator $F(u)+a(t)u$ 
is injective.
Thus, Assumption A) implies that the operator $F(u)+a(t)u$ is injective 
and surjective, it is continuously Fr\'echet differentiable,
as well as its inverse, so the map $u\mapsto F(u)+a(t)u$ is a 
diffeomorphism.
Therefore equation \eqref{eq1.5} is uniquely solvable for $u$ for any $\psi$, 
and the map $\psi=\psi(u)$ is a diffeomorphism. 
The inverse map $u=g(\psi)$, is continuously differentiable 
by the inverse function theorem 
since the operator $\psi'_u=F'(u) + a(t)I$ is boundedly invertible if $a(t)>0$. 
Recall that $F'(u)\ge 0$, because $F$ is monotone. If 
$a(t)\in C^1([0,\infty))$ 
then the solution $u=u(t)$ of equation \eqref{eq1.5} 
is continuously differentiable with respect to $t$ 
(see \cite{R499}, p. 260-261), and if $u=u(t)$ is continuously 
differentiable with respect to $t$, 
then so is $\psi=\psi(u(t))$. 
The differentiability of $u=u(\psi(t))$ also follows from a consequence of the classical inverse function theorem 
(see, e.g., \cite{D}, Corollary 15.1, p. 150). 
Therefore, equation \eqref{eq1} can be 
written in an equivalent form as
\begin{equation}
\label{eq1.6}
\dot{\psi}(t) =  \dot{a}(t)u(t) - \psi(t):= G(t,\psi),\qquad \psi(0):=\psi(u_0).
\end{equation}
Since $u(t)=g(\psi(t))$ is continuously differentiable with respect to $\psi$, 
the map $\psi\mapsto G(t,\psi)$ is continuously differentiable with respect to $\psi$. Therefore, this map
 is Lipschitz, and local existence of the solution to problem \eqref{eq1.6} follows from the standard 
 result (see, e.g., \cite{R499}, p.247). 
Since the map $g(\psi)$ is continuously differentiable and $\dot{\psi}$ is a continuous function 
of $t$, the function $\dot{u}$ is a continuous function of $t$, and problem \eqref{eq1.6} is 
equivalent to problem \eqref{eq1}. We have proved the following theorem
\begin{theorem}
If Assumption A) holds, then problem \eqref{eq1} has a unique local 
solution. 
\end{theorem}

In Section 2 we discuss existence of the global solution to problem \eqref{eq1}.

\section{Existence of the global solution}
\label{sec2}

Since $G(t,\psi)$ is Lipschitz with respect to $\psi$ and 
continuously differentiable 
with respect to $t$, the solution to \eqref{eq1.6} exists globally,
i.e., for all $t\ge 0$, if
\begin{equation}
\label{eq2.1}
\sup_{t\ge 0}\|\psi(t)\| \le c<\infty.
\end{equation}
If the solution $\psi$ to problem \eqref{eq1.6} exists globally, 
then the solution $u(t)$ to the equivalent problem \eqref{eq1} 
exists globally because the map 
$\psi\mapsto u$ is a diffeomorphism.

Let us prove \eqref{eq2.1} assuming that
\begin{equation}
\label{eq2.2}
0<a(t)<C,\quad \frac{1}{2}>\frac{|\dot{a}(t)|}{a(t)},\qquad t \ge 0,
\end{equation}
where $C>$ is a constant.

Denote $h(t):=\|\psi(t)\|$. 
Multiply both sides of \eqref{eq1.6} with $\psi(t)$ and get
\begin{equation}
\label{eq2.3}
h\dot{h} = -h^2 + \langle \dot{a}(t)u(t),\psi\rangle.
\end{equation}

Let $w(t)$ solve the equation:
\begin{equation}
\label{eq2.4}
F(w(t)) + a(t)w(t) -f =0,\qquad t\ge 0.
\end{equation}
Equation \eqref{eq2.3} implies 
\begin{equation}
\label{eq2.5}
\begin{split}
\dot{h} \le  -h + \|\dot{a}|\|u(t)-w(t)\| + |\dot{a}(t)|\|w(t)\|.
\end{split}
\end{equation}
We will prove later the following estimate
\begin{equation}
\label{eq2.6}
\|u(t) - w(t)\| \le \frac{h(t)}{a(t)},\qquad \forall t\ge 0. 
\end{equation}
If \eqref{eq2.6} holds, then \eqref{eq2.5} implies
\begin{equation}
\label{eq2.7}
\dot{h} \le - h\bigg{(}1 -\frac{|\dot{a}(t)|}{a(t)}\bigg{)} + |\dot{a}|\|w(t)\|
\le -\frac{h}{2} + |\dot{a}(t)|\|w\|.
\end{equation}
Therefore,
\begin{equation}
\label{eq2.8}
h(t) \le h(0) e^{-\frac{t}{2}} + e^{-\frac{t}{2}} \int_0^t e^{\frac{s}{2}}|\dot{a}(s)|\|w(s)\|ds,\qquad
\forall t\ge 0.
\end{equation}
From \eqref{eq2.8} and \eqref{eq2.2} one gets
\begin{equation}
\label{eq2.9}
h(t) \le h(0) e^{-\frac{t}{2}} + e^{-\frac{t}{2}} \int_0^t 
e^{\frac{s}{2}}\frac{a(s)}{2}\|w(s)\|ds,\qquad
\forall t\ge 0.
\end{equation}

Let us recall the following result
\begin{lemma}[Lemma 2, \cite{R554}]
\label{lemma2.1}
The function $a\|w_a\|$ is an increasing function of $a$ on $(0,\infty)$, where  
$w_a$ solves \eqref{eq2.4} with $a(t)= a=const>0$.
\end{lemma}

From Lemma~\ref{lemma2.1}, inequality \eqref{eq2.2}, and inequality 
\eqref{eq2.9} one gets
\begin{equation}
\label{eq2.10}
h(t) \le h(0) e^{-\frac{t}{2}} + (1- e^{-\frac{t}{2}})C\|w_C\|,\qquad
\forall t\ge 0.
\end{equation}
Here $w_C$ solves equation \eqref{eq2.4} with $a(t)=C$. 

Therefore, estimate \eqref{eq2.1} is proved as soon as \eqref{eq2.6} is verified. 

Let us state our result and then prove \eqref{eq2.6}.

\begin{theorem}
\label{theorem2.1}
If Assumption A) and \eqref{eq2.2} hold, then problem \eqref{eq1} has a unique global solution. 
\end{theorem}

Let us verify \eqref{eq2.6}.

Using \eqref{eqx12} one gets:
\begin{equation}
\label{eq2.11}
\langle F(u) - F(w) + a(u-w), u-w\rangle \ge 
a\|u-w\|^2.
\end{equation}
Thus,
\begin{equation}
\label{eq2.12}
\|u(t) - w(t)\| \le \frac{\|F(u(t)) - F(w(t)) + a(t)(u(t)-w(t)\|}{a(t)} = \frac{h(t)}{a(t)}.
\end{equation}
So \eqref{eq2.6} is verified and Theorem~\ref{theorem2.1} is proved.\hfill 
$\Box$

\section{Justification of the DSM}

By the justification of the DSM for solving equation
\begin{equation}
\label{eq3.1}
F(y) = f,
\end{equation}
we mean the proof of the following statements (see \cite{R499}, p. 1,
formulas (1.1.5)):
\begin{equation}
\label{eq3.2}
\exists ! u(t),\quad \forall t\ge 0;\quad \exists u(\infty); \quad F(u(\infty)) = f.
\end{equation}

In Theorem~\ref{theorem2.1} the first of these statements is proved. Let us assume
\begin{equation}
\label{eq3.3}
\lim_{t\to \infty}a(t) = 0,\qquad \forall t\ge 0,
\end{equation}
and prove the remaining two statements from \eqref{eq3.2}.

\begin{theorem}
\label{theorem3.1}
If Assumption A), \eqref{eq2.2} and \eqref{eq3.3} hold, and equation \eqref{eq3.1} has a 
solution, then \eqref{eq3.2} hold, and $u(\infty)=y$, where $y$ is the 
unique minimal-norm solution to \eqref{eq3.1}. 
\end{theorem}

\begin{proof}
It is known (see, e.g., \cite{R499}) that
\begin{equation}
\label{eq3.4}
\lim_{t\to\infty} w(t) = y.
\end{equation}
Inequality \eqref{eq2.1} implies
\begin{equation}
\label{eq3.5}
\sup_{t\ge 0}\|u(t)\| \le c,
\end{equation}
because the map $\psi\mapsto u$ is a diffeomorphism. 
Recall that by $c>0$ we denote various constants. 
Inequality \eqref{eq2.3} implies
\begin{equation}
\label{eq3.6}
\dot{h} \le -h + c|\dot{a}(t)|. 
\end{equation}
Assumptions \eqref{eq2.2} and \eqref{eq3.3} imply that
\begin{equation}
\label{eq3.7}
\lim_{t\to\infty} |\dot{a}(t)| = 0.
\end{equation}
From \eqref{eq3.6} and \eqref{eq3.7} it follows that
\begin{equation}
\label{eq3.8}
\lim_{t\to\infty} \psi(t) = 0.
\end{equation}
Indeed, \eqref{eq3.6} implies
\begin{equation}
\label{eq3.9}
h(t) \le h(0)e^{-t} + e^{-t}\int_0^t e^s |\dot{a}(s)|ds,\qquad t\ge 0.
\end{equation}
Clearly $\lim_{t\to\infty} h(0)e^{-t} = 0$, and the L'Hospital rule yeilds
\begin{equation}
\label{eq3.10}
\lim_{t\to\infty} \frac{\int_0^t e^s|\dot{a}(s)|ds}{e^t} = \lim_{t\to\infty}|\dot{a}(t)| = 0.
\end{equation}
Thus, \eqref{eq3.8} is proved. 

Since the map $\psi\mapsto u$ is a diffeomorphism, relation \eqref{eq3.8} 
implies the existence of the limit 
$u(\infty):=\lim_{t\to\infty}u(t)$ and the relation 
\begin{equation}
\label{eq3.11}
F(u(\infty)) = f,
\end{equation}
because $a(\infty) = 0$. 

Finally, $u(\infty) = y$, that is, $u(\infty)$ is the minimal-norm 
solution to equation \eqref{eq3.1}. 
Indeed, $u(\infty)$ is the limit of $u(t)$ where
\begin{equation}
\label{eq3.12}
F(u(t)) + a(t)u(t) -f = 0.
\end{equation}
It is proved in \cite{R499} that the limit, as $a\to 0$, of the solution 
$w$ to the following equation:
\begin{equation}
\label{eq3.13}
F(w_a) + aw_a -f =0,\qquad a>0,
\end{equation}
with a hemicontinuous monotone operator $F$
is the minimal-norm solution to the equation \eqref{eq3.1}, provided 
that equation \eqref{eq3.1} is solvable. 

Theorem~\ref{theorem3.1} is proved. 
\end{proof}

\bibliographystyle{amsplain}

\end{document}